\documentclass[11pt]{article}

\usepackage{amssymb,latexsym}
\usepackage{graphicx,amsmath,amsfonts,amsthm}
\usepackage{apacite}
\usepackage{color,soul}
\usepackage{comment}

\def\row#1#2{{#1}_1,\ldots ,{#1}_{#2}}

\def\row#1#2{{#1}_1,\ldots ,{#1}_{#2}}

\def\2vec#1#2{\left(\begin{array}{c}{#1}\\{#2}\end{array}\right)}
\newtheorem{theorem}{Theorem}
\newtheorem{corollary}{Corollary}
\newtheorem{lemma}{Lemma}

\newtheorem{proposition}{Proposition}
\newtheorem{example}{Example}

\newtheorem{definition}{Definition}

\begin{document}

\setcounter{page}{0}

\thispagestyle{empty}

\begin{center}
{\Large \bf A Characterization of Ideal Weighted Secret Sharing Schemes}
\end{center}
\bigskip

\begin{center}
{\large \bf Ali Hameed and Arkadii Slinko}
\end{center}

\bigskip

\bigskip
\bigskip
\noindent
\begin{center}
\begin{minipage}{14cm}
{\bf Abstract.}  Beimel, Tassa and Weinreb (2008) and Farras and Padro (2010) partially characterized access structures of ideal weighted threshold secret sharing schemes in terms of the operation of composition. They classified indecomposable ideal weighted threshold access structures, and proved that any other ideal weighted threshold access structure is a composition of indecomposable  ones. It remained unclear  which compositions of indecomposable weighted threshold access structures are weighted. In this paper we fill the gap. Using game-theoretic techniques we determine which compositions of indecomposable ideal access structures are weighted, and obtain an if and only if characterization of ideal weighted threshold secret sharing schemes. 
\end{minipage}
\end{center}
\bigskip

\section{Introduction}

Secret sharing schemes are modifications of cooperative games to the situation when not money but information is shared. Instead of dividing a certain sum of money between participants a secret sharing scheme divides a secret into shares---which is then distributed among participants---so that some coalitions of participants have enough information to recover the secret (authorised coalitions) and some (nonauthorised coalitions) do not.  A scheme is perfect if it gives no information to nonauthorised coalitions whatsoever. A perfect scheme is most informationally efficient if the shares contain the same number of bits as the secret \cite{Karnin83}; such schemes are called ideal. The set of authorised coalitions is said to be the access structure.

However, not all access structures can can carry an ideal secret sharing scheme \cite{Stinson:1992}. Finding a description of those which can carry appeared to be quite difficult. A major milestone in this direction was the paper by \cite{BD91} who showed that all ideal secret sharing schemes can be obtained from matroids. Not all matroids, however, define ideal schemes \cite{Seymour1992} so the problem is reduced to classifying those matroids that do. There was little further progress, if any, in this direction. 

Several authors attempted to classify all ideal access structures in subclasses of secret sharing schemes. These include access structures defined by graphs \cite{BD91}, weighted threshold access structures \cite{beimel:360, padro:2010}, hierarchical access structures \cite{padro:2010}, bipartite and tripartite access structures \cite{Padro:1998, PadroS04,FMP2012}. While in the classes of bipartite and tripartite access structures the ideal ones were given explicitly, for the case of weighted threshold access structures \cite{beimel:360} suggested a new kind of description. This method uses the operation of composition of access structures \cite{martin:j:new-sss-from-old}. The idea is that sometimes all players can be classified into 'strong' players and 'weak' players and the access structure can be decomposed into the main game that contains strong players and the auxiliary game which contains weak players.  Under this approach the first task is obtaining a characterisation of indecomposable structures. \citeA{beimel:360} proved that  every ideal indecomposable secret sharing scheme is either disjunctive hierarchical or tripartite. \citeA{padro:2010,FarrasP12} later gave a more precise classification which was complete (but some access structures that they viewed as indecomposable later appeared to be decomposable). 

If a composition of two weighted access structures were again a weighted structure there will not be need to do anything else. However, we will show that this is not true.  Since the composition of two weighted access structures may not be again weighted, it is not clear which indecomposable structures and in which numbers can be combined to obtain more complex weighted access structures. To answer this question in this paper we undertake a thorough investigation of the operation of composition. \par\medskip

Since the access structure of any secret sharing scheme is a simple game in the sense of \citeA{vNM:b:theoryofgames}, we found it more convenient to use game-theoretic methods and terminology. 

Section~2 of the paper gives the background in simple games. We  introduce some important concepts from game theory like Isbel's desirability relation on players, which will play in this paper an important role. We remind the reader of the concept of complete simple game which is a simple game for which Isbel's desirability relation is complete\footnote{In \cite{padro:2010} such games are called hierarchical.}. We introduce the technique of trading transforms and certificates of nonweightedness \cite{GS2011} for proving that a simple game is a weighted threshold games. 

In Section 3, we give the motivation for the concept of composition $C=G\circ_g H$ of two games $G$ and $H$ over an element $g\in G$, give the definition and examples. The essence of this construction is as follows: in the first game $G$ we choose an element $g\in G$ and replace it with the second game $H$. The winning coalitions in the new game are of two types. Firstly, every winning coalition in $G$ that does not contain $g$ remains winning in $C$. A winning coalition in $G$ which contained $g$ needs a winning coalition of $H$ to be added to it to become winning in $C$. We prove several properties of this operation, in particular, we prove that the operation of composition of games is associative. 

Section~4 presents preliminary results regarding the compositions of ideal games and weighted games in general. We start with reminding the reader  that the composition of two games is ideal if and only if the two games being composed are ideal \cite{beimel:360}. Then we show that if a weighted game is composed of two games, then the two composed games are also weighted. Finally, we prove the first sufficient condition for a composition to be weighted.

Section~5 is devoted to compositions in the class of complete games. We prove that, with few possible exceptions, the composition of two complete games is complete if and only if the composition is over the weakest player relative to the desirability relation of the first game. We show that the composition of two weighted threshold simple games may not be weighted threshold even if we compose over the weakest player. We give some sufficient conditions for the composition of two weighted games to be weighted. 

In Section~6 we prove that onepartite games are indecomposable, and also prove the uniqueness of some decompositions.

In Section~7 we recap the classification of indecomposable ideal weighted simple games given by \citeA{padro:2010}. According to it all ideal indecomposable games are either $k$-out-of-$n$ games or belong to one of the six classes: $\bf B_1$, $\bf B_2$, $\bf B_3$, $\bf T_1$, $\bf T_2$, $\bf T_3$. We show that some of the games in their list are in fact decomposable, and hence arrive at a refined list of all indecomposable ideal weighted simple games.

In Section~8 we investigate which of the games from the refined list can be composed to obtain a new ideal weighted simple game. The result is quite striking; the composition of two indecomposable weighted games is weighted only in two cases: when the first game is a $k$-out-of-$n$ game, or if the first game is of type $\bf B_2$ (from the Farras and Padro list) 
and the second game is an anti-unanimity game where all players are passers i.e., players that can win without forming a coalition with other players. This has a major implication for the refinement of Beimel-Tassa-Weinreb-Farras-Padro theorem.

In Section~9, using the results of Section~8, we show that a game $G$ is an ideal weighted simple game if and only if it  is a composition
\[
G=H_1\circ  \cdots  \circ  H_s\circ  I\circ A_n,
\]
where $H_i$ is a $k_i$-out-of-$n_i$ game for each $i=1,2,\ldots, s$, $A_n$ is an anti-unanimity game, and $I$ is an indecomposable game of types $\bf B_1$, $\bf B_2$, $\bf B_3$, $\bf T_1$, and $\bf T_{3}$. Any of these may be absent but $A_n$ may appear only if $I$ is of type $\bf B_2$. 
The main surprise in this result is that in the decomposition there may be at most one game of types $\bf B_1$, $\bf B_2$, $\bf B_3$, $\bf T_1$, $\bf T_3$.

\section{Preliminaries}

\subsection{Simple Games}

The main motivation for this work comes from secret sharing. However, the access structure on the set of users  is a {\em simple game} on that set so we will use game-theoretic terminology.

\begin{definition}[von Neumann \& Morgenstern, 1944]
A simple game is a pair $G=(P_G,W_G)$, where $P_G$ is a set of players and $W_G\subseteq 2^{P_G}$ is a nonempty set of coalitions which satisfies the monotonicity condition:
\[
\text{if $X\in W_G$ and $X\subseteq Y$, then $Y\in W_G$}.
\]
Coalitions from set $W_G$ are called {\em winning} coalitions of $G$, the remaining ones are called {\em losing}. 
\end{definition}

A typical example of a simple game is the United Nations Security Council, which consists of five permanent members and 10 nonpermanent. 
The passage of a resolution requires that all five permanent members vote for it, and also at least nine members in total. The book by \citeA{tz:b:simplegames} gives many other interesting examples.

A simple game will be called just a game. The set $W_G$ of winning coalitions of a game $G$ is completely determined by the set $W_G^{\text{min}} $ of its minimal winning coalitions. A player which does not belong to any minimal winning coalitions is called a {\em dummy}. He can be removed from any winning coalition without making it losing. A player who is contained in every minimal winning coalition is called a {\em vetoer}. A game with a unique minimal winning coalition is called an {\em oligarchy}. In an oligarchy every player is either a vetoer or a dummy. A player who alone forms a winning coalition is called a {\em passer}. A game in which all minimal winning coalitions are singletons is called {\em anti-oligarchy}. In an anti-oligarchy every player is either a passer or a dummy.

\begin{definition}
A simple game $G$ is called {\em weighted threshold game}  if there exist nonnegative weights $\row wn$ and a real number $q$, called {\em quota}, such that 
\begin{equation}
\label{WMG}
X\in W_G \Longleftrightarrow \sum_{i\in X}w_i\ge q.
\end{equation}
This game is denoted $[q;\row wn]$. We call such a game simply {\em weighted}.
\end{definition}
It is easy to see that the United Nation Security Council can be defined in terms of weights as $[39; 7,\ldots,7,1,\ldots,1]$.
In secret sharing weighted threshold access structures were introduced by \cite{shamir:1979,Blakley1979}.\par\medskip

For $X \subset P$ we will
denote its complement $P \setminus X$ by $X^c$. 

\begin{definition}
Let $G=(P,W)$ be a simple game and $A\subseteq P$. Let us define subsets 
\[
W_{\text{sg}}=\{X\subseteq  A^c\mid X\in W\}, \quad  
W_{\text{rg}}=\{X\subseteq A^c\mid X\cup A\in W\}.
\]
Then the game $G_A=(A^c,W_\text{sg})$ is called a {\em subgame} of $G$ and $G^A=(A^c,W_\text{rg})$ is called a {\em reduced game} of $G$. 
\end{definition}



The two main concepts of the theory of games that we will need here are as follows. 

Given a simple game $G$ on the set of players $P$ we define a relation $\succeq_G$ on $P$ by setting $i \succeq_G j$ if for every set $X\subseteq P$ not containing $i$ and~$j$ 
\begin{equation}
\label{condition}
X\cup \{j\}\in W_G \Longrightarrow X\cup \{i\} \in W_G.
\end{equation}
In such case we will say that $i$ is at least as {\em desirable} (as a coalition partner) as $j$.  In the United Nations Security Council every permanent member will be more desirable than any nonpermanent one.
This relation is reflexive and transitive but not always complete (total) (e.g., see \citeA{CF:j:complete}). The corresponding equivalence relation on $[n]$ will be denoted $\sim_{G} $ and the strict desirability relation as $\succ_G$. If this can cause no confusion we will omit the subscript $G$. 

\begin{definition}
Any game with complete desirability relation is called {\em complete}. 
\end{definition}
\begin{example}
Any weighted game is complete.
\end{example}

We note that in \eqref{condition} we can choose $X$ which is minimal with this property in which case $X\cup\{i\}$ will be a minimal winning coalition. Hence the following is true.

\begin{proposition}
Given a simple game $G$ on the set of players $P$ and two players $i.j\in P$, the relation $i\succ_G j$ is equivalent to the existence of a minimal winning coalition $X$ which contains $i$ but not $j$ such that $(X\setminus \{i\})\cup \{j\}$ is losing.
\end{proposition}

We recap that a sequence of coalitions
\begin{equation}
\label{tradingtransform}
{\cal T}=(\row Xj;\row Yj)
\end{equation}
is a trading transform \cite{tz:b:simplegames} if the coalitions $\row Xj$ can be converted into the coalitions $\row Yj$ by rearranging players. This latter condition can also be expressed as 
\[
|\{i:a\in X_i\}| = |\{i:a\in Y_i\}|\qquad \text{for all $a\in P$}.
\]
It is worthwhile to note that while in (\ref{tradingtransform}) we can consider that no $X_i$ coincides with any of $Y_k$, it is perfectly possible that the sequence $\row Xj$ has some terms equal, the sequence  $\row Yj$ can also contain equal terms. 

\citeA{Elgot60} proved (see also \citeA{tz:b:simplegames})  the following fundamental fact.

\begin{theorem}
A game $G$ is a weighted threshold game if for no integer $j$  there exists a trading transform \eqref{tradingtransform} such that all coalitions $\row Xj$ are winning and all $\row Yj$ are losing.
\end{theorem}

Due to this theorem any trading transform \eqref{tradingtransform} where all coalitions $\row Xj$ are winning and all $\row Yj$ are losing is called a {\em certificate of nonweightedness} \cite{GS2011}.

Completeness can also be characterized in terms of trading transforms \cite{tz:b:simplegames}. 

\begin{theorem}
A game $G$ is complete if no certificate of nonweightedness exists of the form
\begin{equation}
\label{certinc}
{\cal T}=(X\cup \{x\}, Y\cup \{y\}; X\cup \{y\}, Y\cup \{x\}). 
\end{equation}
\end{theorem}
We call \eqref{certinc} a {\em certificate of incompleteness}. This theorem says that completeness is equivalent to the impossibility for two winning coalitions to swap two players and become both losing. This latter property is also called {\em swap robustness}.\par\medskip

A complete game $G=(P,W)$ can be compactly represented using multisets. All its players are split into equivalence classes of players of equal desirability. If, say, we have $m$ equivalence classes, i.e., $P=P_1\cup P_2\cup \ldots \cup P_m$ with $|P_i|=n_i$, then we can think that $P$ is the multiset
\[
\{1^{n_1},2^{n_2},\ldots,m^{n_m}\}.
\]
A submultiset $\{1^{\ell_1},2^{\ell_2},\ldots,m^{\ell_m}\}$ will then denote the class of coalitions where $\ell_i$ players come from $P_i$, $i=1,\ldots,m$. All of them are either winning or all losing. We may enumerate classes so  that $1\succ_G 2\succ_G \cdots \succ_G m$. The game with $m$ classes is called {\em $m$-partite}.

If a game $G $ is complete, then we define {\em  shift-minimal}  \cite{CF:j:complete} winning coalitions  as follows. By a {\em shift} we mean a replacement of a player  of a coalition by a less desirable player which did not belong to it.  Formally, given a  coalition $X$, player $p\in X$ and another player $q\notin X$ such that $q\prec_{G}p$, we say that the coalition $ (X\setminus \{p\})\cup \{q\} $ is obtained from $X$ by a {\em shift}. A winning coalition $X$ is {\em shift-minimal} if  every coalition strictly contained in it and every coalition obtained from it by a shift are losing. A complete game is fully defined by its shift-minimal winning coalitions.

\begin{example}[Onepartite games]
Let $H_{n,k}$ be the game where there are $n$ players and it takes $k$ or more to win. Such games are called  {\em $k$-out-of-$n$ games}.  Alternatively they can be characterised as the class of complete 1-partite games, i.e., the games with a single class of equivalent players. The game $H_{n,n}$ is special and is called the {\em unanimity game} on $n$ players. We will denote it as $U_n$. The game $H_{n,1}$ does not have a name in the literature. We will call it {\em anti-unanimity game} and denote $A_n$.
\end{example}

\begin{example}[Bipartite games]
Here we introduce two important types of bipartite games. A hierarchical disjunctive game $H_\exists ({\bf n},{\bf k})$ with ${\bf n}=(n_1,n_2)$ and ${\bf k}=(k_1,k_2)$  on a multiset $P=\{1^{n_1},2^{n_2}\}$ is defined by the set of winning coalitions 
\[
W_\exists = \{ \{1^{\ell_1},2^{\ell_2}\} \mid (\ell_1\ge k_1) \vee (\ell_1+\ell_2\ge k_2)  \},
\]
where $1\le k_1<k_2$, $k_1\le n_1$ and $k_2-k_1 < n_2$. A hierarchical conjunctive game $H_\forall ({\bf n},{\bf k})$ with ${\bf n}=(n_1,n_2)$ and ${\bf k}=(k_1,k_2)$  on a multiset $P=\{1^{n_1},2^{n_2}\}$ is defined by the set of winning coalitions 
\[
W_\forall = \{ \{1^{\ell_1},2^{\ell_2}\} \mid (\ell_1\ge k_1) \wedge (\ell_1+\ell_2\ge k_2)  \},
\]
where $1\le k_1\le k_2$,  $k_1\le n_1$ and $k_2-k_1 < n_2$. In both cases,  if the restrictions on ${\bf n}$ and ${\bf k}$ are not satisfied the game becomes 1-partite  \cite{gha:t:hierarchical}).
\end{example}

\begin{example}[Tripartite games]
\label{ex}
Here we introduce two types of tripartite games. Let ${\bf n}=(n_1,n_2,n_3)$ and ${\bf k}=(k_1,k_2,k_3)$, where $n_1,n_2,n_3$ and $k_1,k_2,k_3$ are positive integers. The game $\Delta_1({\bf n},{\bf k})$ is defined on the multiset $P=\{1^{n_1},2^{n_2},3^{n_3}\}$ with the set of winning coalitions
\[
\{ \{1^{\ell_1},2^{\ell_2},3^{\ell_3}\} \mid (\ell_1\ge k_1) \vee [(\ell_1+\ell_2\ge k_2)\wedge (\ell_1+\ell_2+\ell_3\ge k_3)  \},
\]
where 
\begin{equation}
\label{delta_cond_1}
k_1<k_3,\quad k_2<k_3,\quad n_1 \geq k_1,\quad n_2 >k_2- k_1  \quad \text{and $\quad n_3> k_3-k_2$}.
\end{equation}
These, in particular, imply $n_1+n_2\ge k_2$.\smallskip

The game $\Delta_2({\bf n},{\bf k})$ is for the case when $n_2 \leq k_2 -k_1$, and it is defined on the multiset $P=\{1^{n_1},2^{n_2},3^{n_3}\}$ with the set of winning coalitions
\[
\{ \{1^{\ell_1},2^{\ell_2},3^{\ell_3}\} \mid (\ell_1+\ell_2\ge k_2) \vee [(\ell_1\ge k_1)\wedge (\ell_1+\ell_2+\ell_3\ge k_3)  \}.
\]
where 
\begin{equation}
\label{delta_cond_2}
k_1< k_2<k_3, \quad n_1+n_2\ge k_2,\quad  n_3> k_3-k_2,  \quad \text{and $\quad n_2+n_3> k_3-k_1$}.
\end{equation}
These conditions, in particular, imply $n_1\ge k_1$  and $n_3\ge 2$.

In both cases,  if the restrictions on ${\bf n}$ and ${\bf k}$ are not satisfied the game either contains dummies or becomes 2-partite or even 1-partite (see a justification of this claim in the appendix). 
\end{example}


The games in these three examples play a crucial role in classification of ideal weighted secret sharing schemes  \cite{beimel:360,padro:2010}.

\section{The Operation of Composition of Games}

The most general type of compositions of simple games was defined by \citeA{Shapley62}. We need a very partial case of that concept here, which is in the context of secret sharing, was introduced by \citeA{martin:j:new-sss-from-old}.

\begin{definition}
\label{decompo}
Let $G$ and $H$ be two games defined on disjoint sets of players and $g \in P_{G}$. We define the composition game $C=G\circ_g H$ by defining $P_{C}=(P_{G}\setminus \{g\}) \cup P_{H}$ and
\[
W_{C}= \{X \subseteq P_C\mid X_G \in W_{G} \text{ or $X_G \cup \{g\} \in W_{G}$ and $X_H \in W_{H}$} \},
\]
where $X_G = X \cap P_{G}$  and  $X_H = X \cap P_{H}$. 
\end{definition}


This is a substitution of the game $H$ instead of a single element $g$ of the first game. All winning compositions in $G$ not containing $g$ remain winning in $C$. If a winning coalition of $G$ contained $g$, then it remains winning in $C$ if $g$ is replaced with a winning coalition of $H$. One might imagine that, if a certain issue is voted in $G$, then voters of $H$ are voted first and then their vote is counted in the first game as if it was a vote of player $g$. Such situation appears, for example, if a very experienced expert resigns from a company, they might wish to replace him with a group of experts.

\begin{definition}
A game $G$ is said to be {\em indecomposable} if there does not exist two games $H$ and $K$ and $h\in P_H$ such that $\min(|H|,|K|)>1$ and $G\cong H\circ_h K$. Alternatively, it is called {\em decomposable}. 
\end{definition}

\begin{example}
\label{vetoers}
Let $G=(P,W)$ be a simple game and $A\subseteq P$ be the set of all vetoers in this game. Let $|A|=m$. Then $G\cong U_{m+1}\circ_u G_A$, where $u$ is any player of $U_{m+1}$. So any game with vetoers is decomposable.
\end{example}

\begin{example}
\label{passers}
Let $G=(P,W)$ be a simple game and $A\subseteq P$ be the set of all passers in this game. Let $|A|=m$. Then $G\cong A_{m+1}\circ_a G_A$, where $a$ is any player of $A_{m+1}$. So any game with passers is decomposable.
\end{example}

Suppose $G=(P,W)$ and $G'=(P',W')$ be two games and $\sigma\colon P\to P'$ is a bijection. We say that $\sigma$ is an isomorphism of $G$ and $G'$, and denote this as $G\cong G'$, if $X\in W$ if and only if $\sigma(X)\in W'$.

It is easy to see that if $|H|=1$, then $ H\circ_h K\cong K$ and, if $|K|=1$, then  $H\circ_h K\cong H$.

\begin{proposition}
\label{prop1}
Let $G,H$ be two games defined on the disjoint set of players and $g\in P_G$. Then 
\[
W_{G\circ_g H}^\text{min}=\{X\mid X\in W_G^\text{min} \text{ and $g\notin X$}\}\cup \{X\cup Y \mid \text{$X\cup \{g\}\in W_G^\text{min}$ and $Y\in W_H^\text{min}$ with $g\notin X$}\}.
\]
\end{proposition}

\begin{proof}
Follows directly from the definition.
\end{proof}

\begin{proposition}
Let $G,H,K$ be three games defined on the disjoint set of players and $g\in P_G$, $h\in P_H$. Then
\[
(G\circ_g H)\circ_h K \cong G\circ_g (H\circ_h K), 
\] 
that is the two compositions are isomorphic.
\end{proposition}

\begin{proof}
Let us classify the minimal winning coalitions of the game $(G\circ_g H)\circ_h K$. By Proposition~\ref{prop1} they can be of the following types:
\begin{itemize}
\item $X\in W_G^\text{min}$ with $g\notin X$;
\item $X\cup Y$, where $X\cup \{g\}\in W_G^\text{min}$ and $Y\in W_H^\text{min}$ with $g\notin X$ and $h\notin Y$;
\item $X\cup Y\cup Z$, where $X\cup \{g\}\in W_G^\text{min}$, $Y\cup \{h\}\in W_H^\text{min}$ and $Z\in W_K^\text{min}$  with $g\notin X$ and $h\notin Y$.
\end{itemize}
It is easy to see that the game $G\circ_g (H\circ_h K)$ has exactly the same minimal winning coalitions.
\end{proof}


\begin{proposition}
Let $G,H$ be two games defined on the disjoint set of players. Then  $G\circ_g H$ has no dummies if and only if both $G$ and $H$ have no dummies. 
\end{proposition}

\begin{proof} 
Straightforward.
\end{proof}

\section{Decompositions of  Weighted Games and Ideal Games}

The following result was proved in  \cite{beimel:360} and was a basis for this new type of description. 

\begin{proposition}
\label{splitprop}
Let $C=G\circ_g H$ be a decomposition of a game $C$ into two games $G$ and $H$ over an element $g\in P_G$, which is not dummy. Then, $C$ is ideal if and only if $G$ and $H$ are also ideal.
\end{proposition}

Suppose we have a class of games ${\cal C}$ such that if the composition $G\circ_g H$ belongs to $\cal C$, then both $G$ and $H$ belong to $\cal C$. 
This proposition means that in any class of games $\cal C$ with the above property we may represent any game as a composition of indecomposable ideal games also belonging to $\cal C$. The class of weighted games as the following lemma shows satisfies the above property,   Hence, if we would like to describe ideal games in the class of weighted games we should look at indecomposable weighted games first.

\begin{lemma}
\label{splitlemma}
Let $C=G\circ_g H$ be a decomposition of a game $C$ into two games $G$ and $H$ over an element $g\in P_G$, which is not dummy. Then,   
 if $C$ is weighted, then $G$ and $H$ are weighted.
\end{lemma}

\begin{proof}
Suppose first that $C$ is weighted but $H$ is not. Then we have a certificate of nonweightedness $(\row Uj;\row Vj)$ for the game $H$. Let also $X$ be any minimal winning coalition of $G$ containing $g$ (since $g$ is not a dummy, it exists). Let $X'=X\setminus \{g\}$. Then 
\[
(X'\cup U_1,\ldots, X'\cup U_j; X'\cup V_1,\ldots, X'\cup V_j)
\]
is a certificate of nonweightedness for $C$. Suppose now that $C$ is weighted but $G$ is not. Then let $(\row Xj;\row Yj)$ be a certificate of nonweightedness for $G$ and $W$ be a fixed minimal winning coalition $W$ for $H$. Define
\[
X_i'=
\begin{cases} 
X_i\setminus \{g\}\cup W & \text{if $g\in X_i$}\\
X_i & \text{if $g\notin X_i$}
\end{cases}
\]
and
\[
Y_i'=
\begin{cases} 
Y_i\setminus \{g\}\cup W & \text{if $g\in Y_i$}\\
Y_i & \text{if $g\notin Y_i$}
\end{cases}
\]
Then, since $|\{i\mid g\in X_i\}|=|\{i\mid g\in Y_i\}|$, the following
\[
(X_1',\ldots, X_j'; Y_1',\ldots, Y_j')
\]
 is a trading transform in $C$.
Moreover, it is a certificate of nonweightedness for $C$ since all $X_1',\ldots, X_j';$ are winning in $C$ and all $Y_1',\ldots, Y_j'$ are losing in $C$. So both assumptions are impossible.
\end{proof}

\begin{corollary}
Every weighted game is a composition of indecomposable weighted games.\footnote{As usual we assume that if a game $G$ is indecomposable, its decomposition into a composition of indecomposable games is $G=G$, i.e., trivial.}
\end{corollary}

The converse is however not true. As we will see in the next section, the composition $C=G\circ_g H$  of two weighted games $G$ and $H$ is seldom weighted. Thus we will pay attention to those cases where compositions are weighted. One of those which we will now consider is when $G$ is a $k$-out-of-$n$ game. In this case all players of $G$ are equivalent and we will often omit $g$ and write the composition as $C=G\circ H$.

\begin{theorem}
Let $H=H_{n,k}$ be a $k$-out-of-$n$ game and G is a weighted simple game. Then 
$C=H\circ G$ is also a weighted game.
\end{theorem}

\begin{proof} 
Let $\row Xm$ be winning and $\row Ym$ be losing coalitions of $C$ such that 
\[
(\row Xm;\row Ym)
\]
is a trading transform. Without loss of generality we may assume that $\row Xm$ are minimal winning coalitions. Let $U_i=X_i\cap H$, then $U_i$ is either winning in $H$ or winning with $h$, hence $|U_i|=k$ or $|U_i|=k-1$. If for a single $i$ we had $|U_i|=k$, then all of the sets $\row Ym$ could not be losing since at least one of them would contain $k$ elements from $H$. Thus $|U_i|=k-1$ for all $i$. In this case we have $X_i=U_i\cup S_i$, where $S_i$ is winning in $G$. Let $Y_i=V_i\cup T_i$, where $V_i\subseteq H$ and $T_i\subseteq G$. Since all coalitions $\row Ym$ are losing in $C$, we get $|V_i|=k-1$ which implies that all $T_i$ are losing in $G$. But now we have obtained a trading transform $(\row Sm;\row Tm)$ in $G$ such that all $S_i$ are winning and all $T_i$ are losing. This contradicts to $G$ being weighted. 
\end{proof}

\section{Compositions of complete games}

We will start with the following observation. It says that if $g\in P_G$ is not the least desirable player of $G$, then the composition $G\circ_g H$ is almost never swap robust, hence is almost never complete.

\begin{lemma}
\label{not_complete}
Let $G,H$ be two games on disjoint sets of players and $H$ is neither a unanimity nor an anti-unanimity. If for two elements $g,g'\in P_G$ we have $g \succ g'$ and $g'$ is not a dummy, then $G\circ_g H$ is not complete.
\end{lemma}

\begin{proof}
As $g$ is more desirable than $g'$, there exists a coalition $X\subseteq P_G$, containing neither $g$ nor $g'$ such that $X\cup \{g\}\in W_G$ and $X\cup \{g'\}\notin W_G$. We may take $X$ to be minimal with this property, then $X\cup \{g\}$ is a minimal winning coalition of $G$. Since $g'$ is not dummy, there exist a  minimal winning coalition $Y$ containing $g'$. The coalition $Y$ may contain $g$ or may not. Firstly, assume that it does contain $g$. Since $H$ is not an oligarchy there exist two distinct winning coalitions of $H$, say $Z_1$ and $Z_2$. Then we can find $z\in Z_1\setminus Z_2$. Then the coalitions $U_1=X\cup Z_1$ and $U_2=(Y\setminus \{g\})\cup Z_2$ are winning in $G\circ_g H$ and coalitions $V_1=(X\cup \{g'\})\cup (Z_1\setminus \{z\})$ and  $V_2=Y\setminus \{g,g'\}\cup (Z_2\cup \{z\})$ are losing in this game since $Z_1\setminus \{z\}$ is losing in $H$ and $Y\setminus \{g'\} = Y\setminus \{g,g'\} \cup \{g\}$ is losing in $G$. Since $V_1$ and $V_2$ are obtained when $U_1$ and $U_2$ swap players $z$ and $g'$, the sequence of sets
$
(U_1,U_2;V_1,V_2)
$
is a certificate of incompleteness for $G\circ_g H$. 

Suppose now $Y$ does not contain $g$. Let $Z$ be any minimal winning coalition of $H$ that has more than one player (it exists since $H$ is not an anti-oligarchy). Let $z\in Z$. Then 
\[
(X\cup Z, Y; X\cup\{g'\}\cup (Z\setminus \{z\}), Y\setminus \{g'\}\cup \{z\})
\]
is a certificate of incompleteness for $G\circ_g H$. 
\end{proof}

This lemma shows that if a composition $G\circ_g H$ of two weighted games is weighted, then almost always $g$ is one of the least desirable players of $G$. The converse as we will see in Section~\ref{inde} is not true. If we compose two weighted games  over the weakest player of the first game, the result will be always complete but not always weighted.  

\begin{theorem}
\label{threecases}
Let $G$ and $H$ be two complete games, $g\in G$ be one of the least desirable players in $G$ but not a dummy. Then for the game $C=G\circ_gH$
\begin{enumerate}
\item[(i)] for $x,y\in P_G\setminus \{g\}$ it holds that $x\succeq_Gy$ if and only if $x\succeq_Cy$. Moreover, $x\succ_Gy$ if and only if $x\succ_Cy$;
\item[(ii)] for $x,y\in P_H$ it holds that $x\succeq_Hy$ if and only if $x\succeq_Cy$. Moreover, $x\succ_Hy$ if and only if $x\succ_Cy$;
\item[(iii)] for $x\in P_G\setminus \{g\}$ and $y\in P_H$, then $x\succeq_Cy$;  if $y$ is not a passer or vetoer in $H$, then $x\succ_Cy$.
\end{enumerate}
In particular, $C$ is complete.
\end{theorem}

\begin{proof}
(i) Suppose $x\succeq_Gy$ but not $x\succeq_Cy$. Then there exist $Z\subseteq C$ such that $Z\cup \{y\}\in W_C$ but  $Z\cup \{x\}\notin W_C$. We can take $Z$ minimal with this property. Consider $Z'=Z\cap P_G$. Then either $Z'\cup \{y\}$ is winning in $G$, or else $Z'\cup \{y\}$ is losing in $G$ but $Z'\cup \{y\}\cup \{g\}$ is winning in $G$. In the latter case $Z\cap P_H\in W_H$. In the first case, since $x \succeq_G y$, we have also $Z'\cup \{x\}\in W_G$, which contradicts $Z\cup \{x\}\notin W_C$. Similarly, in the second case we have $Z'\cup \{x\}\cup \{g\}\in W_G$ and since $Z\cap P_H\in W_H$, this contradicts $Z\cup \{x\}\notin W_C$ also. Hence $x\succeq_Cy$.\par

If $x\succ_Gy$, then there exists $S\subseteq P_G$ such that $S\cap \{x,y\}=\emptyset$ and $S\cup \{x\} \in W_G$ but $S\cup \{x\} \notin W_G$. We may assume $S$ is minimal with this property.  If $S$ does not contain $g$, then $S$ is also winning in $C$ and $x\succ_Cy$, so we are done. 
(ii) This case is similar to the previous one. If $S$ contains $g$, then consider any winning coalition $K$ in $H$. Then $(S\setminus \{g\})\cup \{x\}\cup K$ is winning in $C$ whille $(S\setminus \{g\})\cup \{y\}\cup K$ is  losing in $C$. Hence $x\succ_Cy$.

(iii) We have $x\succeq_Gg$ since $g$ is from the least desirable class in $G$. Let us consider a coalition $Z\subset C$ such that $Z\cap \{x,y\}=\emptyset$, and suppose there exists $Z\cup \{y\}\in W_C$ but $Z\cup \{x\}\notin W_C$. Then $Z$ must be losing in $C$, and hence $Z\cap P_G$ cannot be winning in $G$, but $Z\cap P_G\cup \{g\}$ must be winning in $G$. However, since $x\succeq_Gg$, the coalition $Z\cap P_G\cup \{x\}$ is also winning in $G$. But then $Z\cup \{x\}$ is winning in $C$, a contradiction. This shows that if $Z \cup \{y\}$ is winning in $C$, then $Z \cup \{x\}$ is also winning in $C$, meaning $x \succeq_C y$. Thus $C$ is a complete game.

Moreover, suppose that $y$ is not a passer or a vetoer in $H$, we will show that $x \succ_C y$. Since $g$ is not a dummy, then $x$ is not a dummy either. Let $X$ be a minimal winning coalition of $G$ containing $x$. If $g\notin X$, then $X$ is also winning in $C$. However, $X\setminus \{x\}\cup \{y\}$ is losing in $C$, since $y$ is not a passer in $H$. Thus it is not true that $y \succeq_C x$ in this case. If $g\in X$, then consider a winning coalition $Y$ in $H$ not containing $y$ (this is possible since $y$ is not a vetoer in $H$). Then $X\setminus \{g\}\cup Y\in W_C$ but 
\[
X\setminus \{x\}\cup \{g\}\cup \{y\}\cup Y\notin W_C,
\]
whence it is not true that $y \succeq_C x$ in this case as well. Thus $x \succ_C y$ in case $y$ is neither a passer nor a vetoer in $H$.
\end{proof}

\section{Indecomposable onepartite games and uniqueness of some decompositions}

\begin{theorem}
A game $H_{n,k}$ for $n\ne k \ne 1$ is indecomposable.
\end{theorem}

\begin{proof}
Suppose $H_{n,k}$ is decomposable into $H_{n,k} = K \circ_g L$, where $K=(P_K,W_K), L=(P_L,W_L)$ with $n_1=|P_K| \ge 2$ and $n_2=|P_L|\ge 2$. If $g$ is a passer in $K$, then it is the only passer, otherwise if there is another passer $g'$ in $K$, then $\{g'\}$ is winning in the composition, contradicting $k \ne 1$. 
\par
We will firstly show that $n_2 < k$. Suppose that $n_2 \geq k$, and choose a player $h \in P_K$ different from $g$. Consider a coalition $X$ containing $k$ players from $P_L$, then $X$ is winning in the composition and $g$ is a passer, and it is also true that $X$ is a minimal winning coalition in $L$. Now replace a player $x$ in $X$ from $P_L$ with $h$. The resulting coalition, although it has $k$ players, is losing in the composition, because $x$ is not a passer in $K$, and $k-1$ players from $P_L$ are losing in $L$. Therefore $k > n_2$.
\par
We also have $|P_K \setminus \{g\}|=n-n_2 > k -n_2>0$. Let us choose any coalition $Z$ in $P_K \setminus \{g\}$ with $k - n_2$ players. Note that it does not win with $g$ as $|Z \cup \{g\}|= k - n_2 + 1 < k$ players. This is why $Z \cup P_L$ is also losing despite having $k$ players in total, contradiction.
\end{proof}

If the first component of the composition is a $k$-out-of-$n$ game, there is a uniqueness of decomposition.

\begin{theorem}
\label{uniH}
Let $H_{n_1,k_1}$ and $H_{n_2,k_2}$ be two $k$-out-of-$n$ games which are not unanimity games. Then, if $G=H_{n_1,k_1}\circ G_1= H_{n_2,k_2}\circ G_2$, with $G_1$ and $G_2$ having no passers,  then $n_1=n_2$, $k_1=k_2$ and $G_1= G_2$. If $G=U_{n_1}\circ  G_1=U_{n_2}\circ G_2$ and $G_1$ and $G_2$ does not have vetoers, then $n_1=n_2$ and $G_1=G_2$.
\end{theorem}

\begin{proof}
Suppose that we know that $G=H\circ G_1$, where $H$ is a $k$-out-of-$n$ game but not a unanimity game. Then all winning coalitions in $G$ of smallest cardinality have $k$ players, so $k$ in this case can be recovered unambiguously. 

If $G_1$ does not have passers, then $n$ can be also recovered since the set of all players that participate in winning coalitions of size $k$ will have cardinality $n-1$. So there cannot exist two decompositions $G=H_{n_1,k_1}\circ G_1$ and $G=H_{n_2,k_2}\circ G_2$ of $G$, where $k_1\ne k_2$ with $k_1\ne n_1$ and $k_2\ne n_2$. 

Let us consider now the game $G=U\circ G_1$, where $U$ is a unanimity game. Due to Example~\ref{vetoers} if $G_1$ does not have vetoers, then $U$ consists of all vetoers of $G$ and uniquely recoverable.
\end{proof}

\section{Indecomposable Ideal Weighted Simple Games}
\label{inde}

The following theorem was proved in~\cite[p.234]{padro:2010} and will be of a major importance in this chapter.

\begin{theorem}[Farr\`{a}s-Padr\'{o}, 2010]
\label{FP2010}
Any indecomposable ideal weighted simple game belongs to one of the seven following types:
\begin{description}
\item[{\bf H}:] Simple majority or $k$-out-of-$n$ games.

\item[\textbf{B$_1$}:]  Hierarchical conjunctive games $H_\forall(n,k)$ with $\textbf{n}=(n_1,n_2)$, $\textbf{k} = (k_1,k_2)$, where $k_1 < n_1$ and $k_2 - k_1 = n_2 - 1 > 0$. Such games have the only shift-minimal winning coalition $\{1^{k_1},2^{k_2-k_1}\}$.

\item[\textbf{B$_2$}:] Hierarchical disjunctive games $H_\exists(n,k)$ with $\textbf{n}=(n_1,n_2), \textbf{k}=(k_1,k_2)$, where $1 < k_1 \leq n_1$,  $k_2 \leq n_2$, and $k_2=k_1+1$. 
The shift-minimal winning coalitions have the forms $\{1^{k_1}\}$ and $\{2^{k_2}\}$.
\label{p_list_1}
\item[\textbf{B$_3$}:] Hierarchical disjunctive games $H_\exists(n,k)$ with $ \textbf{n}=(n_1,n_2), \textbf{k}=(k_1,k_2)$, where $k_1 \leq n_1$,  $k_2 > n_2 > 2$ and $k_2=k_1+1$.  
The shift-minimal winning coalitions have the forms $\{1^{k_1}\}$ and $\{1^{k_2-n_2},2^{n_2}\}$.


\item[\textbf{T$_{1}$}:]  Tripartite games $\Delta_1({\bf n},{\bf k})$ with $k_1> 1$, $k_2 < n_2$, $k_3=k_1+1 $ and $n_3= k_3-k_2+1 > 2$.  It has two types of shift-minimal winning coalitions: $\{1^{k_1}\}$ and $\{2^{k_2},3^{k_3-k_2}\}$. It follows from \eqref{delta_cond_1} that $k_1\le n_1$ and $k_3-k_2\le n_3$.


\item[\textbf{T$_{2}$}:]  Tripartite games $\Delta_1({\bf n},{\bf k})$ with $n_3= k_3-k_2+1 > 2$ and $k_3=k_1+1$. It has two types of shift-minimal winning coalitions: $\{1^{k_1}\}$ and $\{1^{k_2-n_2},2^{n_2},3^{k_3-k_2}\}$. It follows from \eqref{delta_cond_1} that $k_1\le n_1$,  $k_2-n_2\le k_1$, and $k_3-k_2\le n_3$.


\item[\textbf{T$_{3}$}:] Tripartite games $\Delta_2({\bf n},{\bf k})$ with $k_3-k_1 = n_2+n_3-1$ and $k_3=k_2+ 1$ and $k_2-n_2 > k_1$, $n_3 > 1$. It has two types of shift-minimal winning coalitions  $\{1^{k_2-{n_2}},2^{n_2}\}$ and $\{1^{k_1},2^{k_3-k_1-n_3},3^{n_3}\}$ (the case when $k_3-k_1=n_3$ and $n_2=1$ is not excluded). It follows from \eqref{delta_cond_2} that $k_1\le n_1$, $k_2-n_2\le n_1$, and $k_3-k_1-n_3< n_2$.
\end{description}
\end{theorem}

\noindent Farras and Padro \citeyear{FarrasP12} wrote these families more compactly but equivalently. However, we found it more convenient to use their earlier classification. The list above contains some decomposable games as we will now show.

\begin{proposition}
\label{Prop_B1}
The game of type ${{\bf B}_1}$ for $k_2 - k_1 = n_2 - 1 =1$ is decomposable.
\end{proposition}

\begin{proof}
The decomposition is as follows: Assume $k_2 - k_1 = n_2 - 1 = 1$, so $n_2=2 \ \text{and} \ k_2=k_1+1$, then we have $\textbf{k} = (k_1,k_1+1), \textbf{n}=(n_1,2)$, and the only shift-minimal winning coalition here is $\{1^{k_1},2\}$. Let the first game $G=(P_G,W_G)$, be one-partite with $P_G=\{1^{n_1+1}\}$, $W_G = \{1^{k_1+1}\}$, and let the second game be $H=(P_H,W_H), P_H=\{2^{2}\}, W_H = \{2\}$. Then the composition $G \circ_1 H$ over a player $1 \in P_G$ gives two minimal winning coalitions $\{1^{k_1+1}\}$ and $\{1^{k_1},2\}$, of which only $\{1^{k_1},2\}$ is shift-minimal. Hence the composition is of type ${{\bf B}_1}$. This proves that a game of type \textbf{B$_1$} is decomposable in this case.   
\end{proof}

\begin{proposition}
\label{Prop_UA}
The unanimity games $U_n$ and anti-unanimity $A_n$ for $n>2$ are decomposable. $U_2$ and $A_2$ are indecomposable.
\end{proposition}

\begin{proof}
We note that 
\[
U_n\circ U_m\cong U_{n+m-1}
\]
for any $u\in U_n$. In particular, the only indecomposable unanimity game is $U_2$. Similarly,  
\[
A_n\circ A_m\cong A_{n+m-1}
\]
for any $a\in A_n$ with the only indecomposable anti-unanimity game is $A_2$.
\end{proof}

\begin{proposition}
\label{Prop_T2}
All games of type ${\bf T}_2$ are decomposable. 
\end{proposition}

\begin{proof}
Let $\Delta=\Delta_1({\bf n},{\bf k})$ be of type ${\bf T}_2$. Then we have the following decomposition for it. The first game will be $G=(P_G,W_G)$, which is bipartite with the multiset representation on $ \{1^{n_1},2^{n_2+1}\}$  and shift-minimal winning coalitions of types $\{1^{k_1}\}$ and $ \{1^{k_2-n_2},2^{n_2+1}\}$. The second game will be $(k_3-k_2)$-out-of-$n_3$ game  $H=(P_H,W_H)$, with the multiset representation on $\bar{P}_H = \{3^{n_3}\}$ and shift-minimal winning coalitions of type $\{3^{k_3-k_2}\}$. The composition is over a player $p \in P_G$ from level $2$. Then we can see that $G \circ_p H$ has shift-minimal winning coalitions of types   $\{1^{k_1}\}$ and $\{1^{k_2-n_2},2^{n_2},3^{k_3-k_2}\}$, hence is exactly $\Delta$.
\end{proof}

We now refine classes ${\bf H}$ and ${\bf B}_1$ as follows:
\begin{description}
\item[\textbf{H}:]  $\  \ $Games of this type are $A_2$, $U_2$ and $H_{n,k}$, where $1 < k < n$.

\item[\textbf{B$_1$}:]  Hierarchical conjunctive games $H_\forall(n,k)$ with $\textbf{n}=(n_1,n_2)$, $\textbf{k} = (k_1,k_2)$, where $k_1 < n_1$ and $k_2 - k_1 = n_2 - 1 > 1$.
\end{description}

The following  of Theorem~\ref{FP2010}, is now an if-and-only-if statement.

\begin{theorem} 
\label{list_all}
A game is ideal weighted and indecomposable if and only if it belongs to one of the following types: ${\bf H}, {\bf B}_1, {\bf B}_2, {\bf B}_3, {\bf T}_1, {\bf T}_3$. 
\end{theorem} 

\begin{proof}
Due to Theorem~\ref{FP2010} and Propositions~\ref{Prop_B1}-\ref{Prop_T2} all that remains to show is that the remaining cases are indecomposable. We leave this routine work to the reader.
\end{proof}

Let us compare this theorem with Theorem~\ref{FP2010}.  We narrowed the class ${\bf H}$, we excluded the case $n_2=2$ in ${\bf B}_1$ and removed class ${\bf T}_2$.


\section{Compositions of ideal weighted indecomposable games}

Suppose from now on that we have a composition $G = G_1 \circ_g G_2$, where both $G_1$ and $G_2$ are ideal and weighted, and $G_1$ is indecomposable. The plan now is to fix $G_1$ and analyse what happens when we compose it with an arbitrary ideal weighted game $G_2$. Since $G_1$ is ideal weighted and indecomposable, then it belongs to one of the seven types of games listed in Theorem~\ref{list_all}. So we carry out the analysis case by case for all possibilities of $G_1$. 
\par
The key result that will lead us to the main theorem of this paper is the following.

\begin{theorem}
\label{when}
Let $G$ be a game with no dummies which has a nontrivial decomposition $G= G_1 \circ_g G_2$, such that $G_1$ and $G_2$ are both ideal and weighted, and $G_1$ is indecomposable. Then $G$ is ideal weighted if and only if either 
\begin{itemize}
\item[(i)] $G_1$ is of type $\textbf{H}$, or 
\item[(ii)] $G_1$ is of type \textbf{B}$_2$ and $G_2$ is $A_n$ such that the composition is over a player $g$ of level $2$ of~$G_1$. 
\end{itemize}
\end{theorem}

We will prove it in several steps. Firstly, we will consider all cases when $g$ is from the least desirable level of $G_1$. Secondly, in Appendix, we will deal with the hypothetical cases when $g$ is not from the least desirable level. This is because, unfortunately, Lemma~\ref{not_complete} still leaves a possibility that for some special cases of $G_2$ this decomposition may be over $g$ which is not the least desirable in $G_1$. 

\subsection{The two weighted cases}
\label{1A}

\begin{proposition}
\label{w_1}
If $G_1=(P_1,W_1)$ is of type $\textbf{H}$ and $G_2=(P_2,W_2)$ is weighted, then $G = G_1 \circ_{g} G_2$ is weighted.
\end{proposition} 

\begin{proof}
Assume the contrary. Then $G$ has a certificate of nonweightedness
\[
(X_1,\ldots,X_m; Y_1,\ldots,Y_m),
\]
where  $\row Xm$ are minimal winning coalitions and $\row Ym$ are losing coalitions of $G$. Let $U_i = X_i \cap P_1$, then either $|U_i| = k$ or $|U_i| = k-1$. However, if for a single $i$ we have $|U_i| = k$, then it cannot be that all of the sets $\row Ym$ are losing, as there will be at least one among with at least $k$ elements of $P_1$. Thus $|U_i| = k-1$ for all $i$. In this case we have $X_i = U_i \cup S_i$, where $S_i$ is winning in $G_2$. Let $Y_i = V_i \ \cup \ T_i$, where $V_i \subseteq P_1$ and $T_i \subseteq P_2$. We must have $|V_i| = k-1$ for all $i$. Since all coalitions $\row Ym$ are losing in $G$, then all $T_i$ are losing in $G_2$. But now we have obtained a trading transform $(S_1,\ldots,S_m; T_1,\ldots,T_m)$ for $G_2$, such that all $S_i$ are winning and all $T_i$ are losing in $G_2$, i.e., a certificate of nonweightedness for $G_2$. This contradicts the fact that $G_2$ is weighted.
\end{proof}

\begin{proposition}
\label{g_2}
Let $G_1=(P_1,W_1)$ be a weighted simple game of type \textbf{B}$_2$, $g$ is a player from level $2$ of $P_1$, and $G_2$ is $A_n$, then $G = G_1 \circ_{g} G_2$ is a weighted simple game. 
\end{proposition}
\begin{proof}
Since $g$ is a player from level $2$ of $P_1$, then $G$ is a complete game by Theorem~\ref{threecases}. Also, recall that shift-minimal winning coalitions of a game of type \textbf{B}$_2$ are $\{1^{k_1}\}$ and $\{2^{k_1+1}\}$. We shall prove weightedness of $G$ by showing that it cannot have a certificate of nonweightedness. In the composition, in the multiset notation, $G$ has the following shift-minimal winning coalitions $\{1^{k_1}\},\{2^{k_1},3\}$. So all shift-minimal winning coalitions have $k_1$ players from $P_1 \setminus \{g\}$. Also, since $G_1$ has two thresholds $k_1$ and $k_2$ such that $k_2=k_1+1$, then any coalition containing more than $k_1$ players from $P_1 \setminus \{g\}$ is winning in $G_1$, and hence winning in $G$. Suppose now towards a contradiction that $G$ has the following certificate of nonweightedness
\begin{equation}
\label{baba}
(X_1,\ldots,X_n;Y_1,\ldots,Y_n),
\end{equation}
where $X_1,\ldots,X_n$ are shift-minimal winning coalitions and $Y_1,\ldots,Y_n$ are losing coalitions in $G$. Let the set of players of $A_n$ be $P_{A_n}$. It is easy to see that at least one of the coalitions $X_1,\ldots,X_n$ in~(\ref{baba}) is not of the type $\{1^{k_1}\}$, so at least one of these winning coalitions has a player from the third level, i.e. from $A_n$. But since each shift-minimal winning coalition in~(\ref{baba}) has $k_1$ players from $P_1 \setminus \{g\}$, then each losing coalition $Y_1,\ldots,Y_n$ in~(\ref{baba}) also has $k_1$ players from $P_1 \setminus \{g\}$ (if it has more than $k_1$ then it is winning). Moreover, at least one coalition from $Y_1,\ldots,Y_n$, say $Y_1$, has at least one player from $P_{A_n}$. It follows that $(Y_1 \cap P_{1}) \cup \{g\} \in W_1$ and $Y_1 \cap P_{A_n}$ is winning in $A_n$. Hence $Y_1$ is winning in $G$, contradiction. Therefore no such certificate can exist.
\end{proof}

In the next section we analyse the remaining of compositions $G = G_1 \circ G_2$ in terms of $G_1$, where the composition is over a player from the least desirable level of $G_1$. We will show that none of them is weighted. 


\subsection{All other compositions are nonweighted}
\label{1C}

Here we will consider two cases:
\begin{enumerate}
\item $G_2$ has at least one minimal winning coalition with cardinality at least $2$.
\item $G_2 = A_n$, where $n \geq 2$.
\end{enumerate}
We will start with the following general statement which will help us to resolve the first case.

\begin{definition}
Let $G=(P,W)$ be a simple game and $g\in P$. We say that a coalition $X$ is $g$-winning if $g\notin X$ and $X\cup \{g\}\in W$. 
\end{definition}

Every winning coalition is of course $g$-winning but not the other way around.

\begin{lemma}
\label{keylemma}
Let $G$ be a game for which there exist coalitions $X_1,X_2,Y_1,Y_2$ such that both $X_1$ and $X_2$ do not contain $g$,
\begin{equation}
\label{X1-Y2}
(X_1,X_2\, ;\, Y_1,Y_2)
\end{equation}
is a trading transform, $X_1$ is winning  $X_2$ is $g$-winning and $Y_1$ and $Y_2$ are losing in $G$. Let also $H$ be a game with a minimal winning coalition $U$ which has at least two elements, then $C=G\circ_gH$ is not weighted.
\end{lemma}

\begin{proof}
If $X_2$ is winning in $G$, then there is nothing to prove since \eqref{X1-Y2} is a certificate of nonweightedness for $C$, suppose not.
Let $U=U_1\cup U_2$, where $U_1$ and $U_2$ are losing in $H$. Then it is easy to check that 
\[
(X_1,X_2\cup U\, ;\, Y_1\cup U_1,Y_2\cup U_2)
\]
is a certificate of nonweightedness for $C$. Indeed, $X_1$ and $X_2\cup U$ are both winning in $C$ and $Y_1\cup U_1$ and $Y_2\cup U_2$ are both losing.
\end{proof}

The only exception in this case is when $H$ consists of passers and dummies. We will have to consider this case separately. 

\begin{lemma}
If $G$ is of type ${\bf B_1}$, ${\bf B_2}$ or ${\bf B_3}$, $g$ is any element of level 2, and $H$ has a minimal winning coalition $X$ which has at least two elements, then $G\circ_gH$ is not weighted.
\end{lemma}

\begin{proof}
Suppose $G$ is of type ${\bf B_1}$. Then let us consider the following trading transform
\[
(\{1^{k_1},2^{k_2-k_1}\}, \{1^{k_1},2^{k_2-k_1-1}\}\,;\, \{1^{k_1-1},2^{k_2-k_1+1}\}, \{1^{k_1+1},2^{k_2-k_1-2}\})
\]
(note that $k_2-k_1+1=n_2$ and $k_1+1\le n_1$ so there is enough capacity in both equivalence classes to make all coalitions involved legitimate). It is easy to check that the first coalition in this sequence is winning, the second is $g$-winning and the remaining two are losing. 
By Lemma~\ref{keylemma} the result holds.

Suppose now $G$ is of type ${\bf B_2}$, then $k_2=k_1+1\le n_2$. Let $k_1=k$. Then we can apply Lemma~\ref{keylemma} to the trading transform 
\[
(\{1^k\},\{2^k\}\, ;\, \{1^{\lfloor \frac{k}{2}\rfloor}, 2^{\lceil\frac{k}{2}\rceil}\}, \{1^{\lceil \frac{k}{2}\rceil}, 2^{\lfloor \frac{k}{2}\rfloor}\}),
\]
where $\{1^k\}$ is winning, $\{2^k\}$ is $g$-winning and the remaining two coalitions are losing.

If $G$ is of type ${\bf B_3}$, then $n_2<k_2=k_1+1$. We again let $k=k_1$. In this case we can apply Lemma~\ref{keylemma} to the trading transform
\[
(\{1^k\}, \{1^{k-2},2^2\}\, ; \, \{1^{k-1},2\}, \{1^{k-1},2\}),
\]
where the first coalition is winning, the second is $g$-winning (we use $n_2\ge 3$ here) and the two remaining coalitions are losing. 
\end{proof}

\begin{lemma}
If $G$ is of type ${\bf T}_1$ or ${\bf T}_{3}$, 
$g$ is any element of level 3, and $H$ has a minimal winning coalition $X$ which has at least two elements, then $C=G\circ_gH$ is not weighted.
\end{lemma}

\begin{proof}
If $G$ is of type ${\bf T}_1$.  Then let us consider the following trading transform
\[
(\{1^{k_1}\}, \{2^{k_2},3^{k_3-k_2-1}\}\, ;\, \{1^{k_1-1}, 2\},  \{1, 2^{k_2-1},3^{k_3-k_2-1}\}).
\]
Lemma~\ref{keylemma} is applicable to it so $C$ is not weighted.

Suppose $G$ is of type ${\bf T}_{3}$.  Then let us consider the following trading transform
\[
(\{1^{k_2-n_2},2^{n_2}\}, \{1^{k_1},2^{n_2-1},3^{n_3-1}\}\, ;\, \{1^{k_2-n_2},2^{n_2-1},3\}, \{1^{k_1},2^{n_2},3^{n_3-2}\}).
\]
Since $n_3>1$ all coalitions exist. Lemma~\ref{keylemma} is now applicable and shows that $C$ is not weighted.  This proves the lemma.
\end{proof}

We will now deal with the second case. Denote players of $A_n$ by $P_{A_n}$.

\begin{proposition}
\label{an}
Let $G_1$ be an ideal weighted indecomposable simple game of types ${\bf B}_1$, ${\bf B}_3$, ${\bf T}_1$, and ${\bf T}_3$,  and $g$ be a player from the least desirable level of $G_1$, then $G=G_1 \circ_g A_n$ is not weighted.
\end{proposition}

\begin{proof}
Let $G_1$ be of type \textbf{B$_1$}. The only shift-minimal winning coalition of $G_1$ is of the form $\{1^{k_1},2^{k_2-k_1}\}$, where $n_1>k_1>0$, $k_2-k_1=n_2-1>1$. 
%
Composing over a player of level $2$ of $G_1$ gives shift-minimal winning coalitions of types $\{1^{k_1},2^{k_2-k_1}\}$ and $\{1^{k_1}, 2^{k_2-k_1-1}, 3\}$.
Thus the game is not weighted due to the following certificate of nonweightedness:
\[
(\{1^{k_1}, 2^{k_2-k_1}\}, \{1^{k_1}, 2^{k_2-k_1-1}, 3\}; \{1^{k_1-1}, 2^{k_2-k_1+1}, 3\}, \{1^{k_1+1}, 2^{k_2-k_1-2}\}).
\]
Since in a game of type \textbf{B$_1$} we have $k_2-k_1+1=n_2$ and $k_1+1\le n_1$, then all the coalitions in this trading transform exist. 
\par

Now consider \textbf{B$_3$}. Its shift-minimal winning coalition have types $\{1^{k_1}\}, \{1^{k_2-n_2},2^{n_2}\}$. Composing over a player of level $2$ of $G_1$ gives  the following types of winning coalitions $\{1^{k_1}\}$, $\{1^{k_2-n_2},2^{n_2-1}, 3\}$ in $G$. The game is not weighted due to the following certificate of nonweightedness:
\[
(\{1^{k_2-n_2},2^{n_2-1}, 3\}, \{1^{k_2-n_2},2^{n_2-1}, 3\}; \{1^{k_2-n_2+1},2^{n_2-2}\}, \{1^{k_2-n_2-1},2^{n_2},3^2\}).
\]
Note that $k_2-n_1+1<k_1 \leq n_1$ and $ n_2 > 2$ in \textbf{B$_3$}, so all the coalitions in this transform exist. 
\par
Now consider \textbf{T$_1$}. Since its levels 2 and 3 form a subgame of type \textbf{B$_1$}, composing it with $A_n$ over a player of level 3, as was proved, will result in a nonweighted game.\par


Let us consider ${\bf T}_{3}$, where the shift-minimal winning coalition are $\{1^{k_2-n_2},2^{n_2}\}$, $\{1^{k_1},2^{k_3-k_1-n_3},3^{n_3}\}$. If we compose over a player of level $3$ of $G_1$, then the resulting game will have shift-minimal coalitions of the following type $\{1^{k_1},2^{k_3-k_1-n_3},3^{n_3-1},4\}$, where now elements of $G_2=A_n$ will form level 4. Then we can show that the composition $G_1\circ G_2$ is not weighted due to the following certificate of nonweightedness:
\[
(\{1^{k_1},2^{k_3-k_1-n_3},3^{n_3-1},4\},\{1^{k_1},2^{k_3-k_1-n_3},3^{n_3-1},4\};  
\]
\[
\{1^{k_1+1},2^{k_3-k_1-n_3},3^{n_3-2}\}, \{1^{k_1-1},2^{k_3-k_1-n_3},3^{n_3},4^2\}).
\]
The coalition $\{1^{k_1+1},2^{k_3-k_1-n_3},3^{n_3-2}\}$ is losing because in ${\bf T}_{3}$ we have $k_3-k_1-n_3=n_2-1$ and also $k_2-n_2 > k_1$, meaning $(k_1+1)+(k_3-k_1-n_3)=k_1+1+n_2-1 \leq k_2-n_2+n_2-1 = k_2-1$ Also in total it contains less than $k_3$ elements. The coalition $ \{1^{k_1-1},2^{k_3-k_1-n_3},3^{n_3},4^2\}$ is easily seen to be losing as well.

Now all that remains for the proof of Theorem~\ref{when} is to consider the cases when $g$ is not from the least desirable level of $G_1$ which may happen only when it is of types ${\bf T}_{1}$ and ${\bf T}_{3}$. These cases are similar to those that have been already considered and we delegate them to the Appendix.
\end{proof}

\section{The Main Theorem}

All previous results combined give us the main theorem:

\begin{theorem}
\label{pad2}
$G$ is an ideal weighted simple game if and only if it is a composition
\begin{equation}
\label{magic}
G = H_1 \circ \ldots \circ H_s \circ I \circ_g A_{n} \ \ (s \geq 0);
\end{equation}
where $H_i$ is an indecomposable game of type \textbf{H} for each $i=1,\ldots,s$. Also, $I$, which is allowed to be absent, is an indecomposable game of types \textbf{B$_1$}, \textbf{B$_2$}, \textbf{B$_3$}, \textbf{T$_1$} and \textbf{T$_{3}$}, and $A_{n}$ is the anti-unanimity game on $n$ players. Moreover, $A_n$ can be present only if $I$ is either absent or it is of type \textbf{B$_2$}; in the latter case the composition $I \circ A_n$ is over a player $g$ of the least desirable level of $I$. Also, the above decomposition is unique.
\end{theorem}
 
\begin{proof}
The following proposition will be useful to show the uniqueness of the decomposition of an ideal weighted game.

\begin{proposition}
\label{uniuni}
Let $H$ be a game of type \textbf{H},  $B$ be a game of type ${\bf B}_2$ with $b$ being a player from  level $2$ of $B$, $G$ be an ideal weighted simple game, and $A_n$ be an anti-unanimity game. Then $H \circ G \ncong B \circ_b A_n$.
\end{proposition}

\begin{proof}
 We note that by Theorem~\ref{threecases} both compositions are complete. Recall that isomorphisms preserve Isbell's desirability relation \cite[]{CF:j:complete}. An isomorphism preserves completeness and maps shift-minimal winning coalitions of a complete game onto shift-minimal winning coalitions of another game.

Let $H=H_{k,n}$. Consider first the composition $H \circ G$. Any minimal winning coalition in this composition will have either $k$ or $k-1$ players from the most desirable level. \par

Now consider $B \circ_b A_n$.  Let the two types of shift-minimal winning coalitions of $B$ are of the forms $\{1^\ell \}$ and $\{2^{\ell +1}\}$, then there will be a minimal winning coalition in $B \circ_b A_n$ which has $\ell $ players from the second most desirable level and an element of level 3 with no players of level 1.  

The two games therefore cannot be isomorphic.  
\end{proof}

\noindent {\it Proof of Theorem~\ref{pad2}.}
This proof is now easy since the main work has been done in Theorem~\ref{when}. Either $G$ is decomposable or not. If it is not, then by Theorem~\ref{list_all} it is either of type ${\bf H}$ or one of the indecomposable games of types \textbf{B$_1$}, \textbf{B$_2$}, \textbf{B$_3$}, \textbf{T$_1$}, and \textbf{T$_{3}$}. So the theorem is trivially true. Suppose now that $G$ is decomposable, so $G = G_1 \circ G_2$. Then by Theorem~\ref{when} there are only two possibilities:
\begin{itemize}
\item[(i)] $G_1$ is of type $\textbf{H}$;
\item[(ii)] $G_1$ is  of type \textbf{B$_2$}, and also $G_2 = A_n$ such that the composition is over a player of  level $2$ of $G_1$.
\end{itemize}

By Proposition~\ref{uniuni} these two cases are mutually exclusive. Suppose we have the case (i). By Theorem~\ref{uniH} $G_1$ is uniquely defined and we can apply the induction hypothesis to $G_2$. It is also easy to see that in the second case $G_1$ and $G_2$ are uniquely defined.
\end{proof}

\section{Acknowledgments} 

Authors thank Carles Padro for a number of useful discussions. We are very grateful to Sascha Kurz for a very useful feedback on the early draft of this paper.

\bibliographystyle{apacite}
\bibliography{short}

\section{Appendix} 

\subsection{A canonical representation of $\Delta_1$ and $\Delta_2$.}

\begin{proposition}
\label{delta1}
The game $\Delta_1({\bf n},{\bf k})$ is tripartite game without dummies if and only if conditions \eqref{delta_cond_1} are satisfied.
\end{proposition}

\begin{proof}
It is easy to see from the definition that this game is complete and $1\succeq_G2\succeq_G 3$. Suppose we actually have $1\succ_G2\succ_G 3$ so that the game is tripartite.  If the condition $k_1\le n_1$ is not satisfied the condition $\ell_1\ge k_1$ has no solution and $1$ becomes equivalent to $2$. So we assume $k_1\le n_1$. If $k_2\ge k_3$, then the condition $ \ell_1+\ell_2\ge k_2$ is redundant which implies $2\sim 3$ and the game is bipartite so we assume $k_2<k_3$.  If $k_1\ge k_3$, then the coalition  $\ell_1+\ell_2+\ell_3\ge k_3$ is redundant and 3 is a dummy. Hence we assume $k_1<k_3$. If we only had $n_2 \le k_2- k_1$, then $\ell_1+\ell_2\ge k_2$ can be satisfied only if $\ell_1\ge k_1$ is satisfied. So in this case $\{1^{k_1}\}$ is the only minimal winning coalition, which implies $2\sim 3$. So $n_2>k_2- k_1$. Finally, if $n_3> k_3-k_2$ is not satisfied, then $\ell_1+\ell_2+\ell_3\ge k_3$ implies $\ell_1+\ell_2\ge k_2$, in which case the minimal winning coalition must satisfy either $\ell_1=k_1$ or $\ell_1+\ell_2+\ell_3= k_3$. We get in this case $2\sim 3$, which is impossible. Hence if $\Delta_1({\bf n},{\bf k})$   is tripartite and has no dummies, the conditions \eqref{delta_cond_1} are satisfied.

On the other hand, if \eqref{delta_cond_1} are satisfied, then the game has two shift-minimal winning coalitions $\{1^{k_1}\}$ and either $\{2^{k_2}, 3^{k_3-k_2}\}$ in case $k_2\le n_2$ or  $\{1^{k_2-n_2}, 2^{n_2}, 3^{k_3-k_2}\}$ in case $k_2> n_2$. In both cases $1\succ 2\succ  3$ by Proposition~1.
\end{proof}

\begin{proposition}
\label{delta2}
The game $\Delta_2({\bf n},{\bf k})$ is tripartite game without dummies if and only if conditions \eqref{delta_cond_2} are satisfied.
\end{proposition}

\begin{proof}
Suppose $\Delta_2({\bf n},{\bf k})$ is tripartite. Like in Proposition~\ref{delta1} we find that $k_1<k_3$ and $k_2<k_3$. However, we also know that $k_2-k_1\ge n_2>0$. Hence we assume $k_1<k_2<k_3$.   If $n_1+n_2\ge k_2$ is not satisfied, then $\ell_1+\ell_2 \ge k_2$ is ineffectual and $2\sim 3$. So we assume $n_1+n_2\ge k_2$. In this case we have a shift-minimal winning coalition $C=\{1^{k_2-n_2}, 2^{n_2}\}$ and secures that $2\succ 3$ (as $k_2<k_3$).  If $n_3> k_3-k_2$ is not satisfied, then  $ \ell_1+\ell_2+\ell_3\ge k_3$ is redundant and $3$ is a dummy.  Since $k_3>k_2$ we have $n_3\ge  k_3-k_2+1\ge 2$. Since $\Delta_2({\bf n},{\bf k})$  is defined for the case $n_2 \leq k_2 -k_1$,  we have  $k_1\le k_2-n_2\le n_1$ and $n_1\ge k_1$ follows.

Now, if the  coalitions $ \{1^{k_1}\}$ and $\{2^{k_3-k_1-n_3+1}\}$ exist, then a replacement of 1 with 2 in a winning coalition
$ \{1^{k_1-1}, 2^{k_3-k_1-n_3+1}, 3^{n_3}\}$ results in a losing coalition $ \{1^{k_1}, 2^{k_3-k_1-n_3}, 3^{n_3}\}$. As the conditions \eqref{delta_cond_2} imply $k_1\le n_1$, the first coalition exists. The second coalition exists since $k_3-k_1-n_3< n_2$ is equivalent to $k_3-k_1<n_2+n_3$. This implies $1\succ 2$.

Now, since $ n_1+n_2\ge k_2$ and $k_2<k_3$, there exists a minimal winning coalition $\{1^{\ell_1},2^{\ell_2}\}$ with $\ell_1+\ell_2=k_2$ and $\ell_2\ge 1$. A replacement of 2 here with a $3$ leads to a losing coalition, hence $2\succ 3$. 
\end{proof}

\subsection{End of proof of Theorem~\ref{when}}

Here we have to deal with the hypothetical possibility that $G$ does not fall into categories (i) and (ii). Then we know that $G_1$ has at least two desirability levels and $g$ is not from the least desirable level. Also Lemma~\ref{not_complete} implies that in this case  $G_2=A_n$ or $G_2=U_n$ for some $n\ge 2$. Let us deal with $G_2=A_n$ first. We need the following

\begin{lemma}
\label{X_1X_2}
Let $G=(P,W)$ be a game where player $g$ is strictly more desirable than player $g'$. Suppose also that we can find two coalitions
$X_1$ and $X_2$ in $G$
such that 
\begin{equation}
\label{eq1}
g'\notin X_1,\quad X_1\cup \{g\}\in W,\quad X_1\cup \{g'\}\in L;
\end{equation}
\begin{equation}
\label{eq2}
g'\in X_2,\quad X_2\cup \{g\}\in W,\quad X_2\setminus \{g'\}\cup \{g\}\in L.
\end{equation}
Then the composition $C=G\circ_gA_n$, $n\ge 2$, is not complete.
\end{lemma}

\begin{proof}
Let $a,b\in A_n$. We have the following certificate of incompleteness:
\[
(X_1\cup \{a\},X_2\cup \{b\};\, X_1\cup \{g'\},X_2\setminus \{g'\}\cup \{a,b\}).
\]
Indeed, both $X_1$ and $X_2$ win with $g$ in $G$ and both $\{a\}$ and $\{b\}$ are winning coalitions in $H$, so $X_1\cup \{a\}$ and $X_2\cup \{b\}$ are winning in $C$. On the other hand $X_1 \cup \{g'\}$ and $X_2\cup \{g'\}$ are losing in $G$ and the latter even losing with $g$ so $X_1\cup \{g'\}$ and $X_2\setminus \{g'\}\cup \{a,b\}$ are both losing in $C$.  This proves the lemma.
\end{proof}

\begin{lemma}
\label{An_certs}
Let $G$ be an indecomposable simple game of one of the types ${\bf B}_1$, ${\bf B}_2$, ${\bf B}_3$, ${\bf T}_{1}$,  and ${\bf T}_{3}$, and let $g$ be a player of $G$ which is not from the least desirable level. Then the composition $G \circ_g A_n$ is not complete for all $n\ge 2$.
\end{lemma}

\begin{proof}
\label{rest_of}
 Let us first consider the case where $g$ is from the most desirable level of $G$. We will apply Lemma~\ref{X_1X_2} to show that $G \circ_g A_n$ is not complete. So in what follows we show that for each case there exists $g,g' \in P$ and coalitions $X_1$ and $X_2$ of $G$ which satisfy the conditions of Lemma~\ref{X_1X_2}.  In the following three cases, $g$ is a player of level $1$ and $g'$ is a player of level $2$. 
\begin{itemize}
\item[(i)] ${\bf B}_1$: $X_1$ is of type $\{1^{k_1-1},2^{k_2-k_1}\}$, and $X_2$ is of type $\{1^{k_1-1},2^{k_2-k_1}\}$;
\item[(ii)] ${\bf B}_2$: $X_1$ is of type $\{1^{k_1-1}\}$, and $X_2$ is of type $\{2^{k_1}\}$;
\item[(iii)] ${\bf B}_3$: $X_1$ is of type $\{1^{k_1-1}\}$, and $X_2$ is of type $\{1^{k_2-n_2},2^{n_2-1}\}$.
\end{itemize}
\par And for the following three cases, $g$ is a player of level $1$ and $g'$ is a player of level $3$. 
\begin{itemize}
\item[(iv)] ${\bf T}_1$: $X_1$ is of type $\{1^{k_1-1}\}$, and $X_2$ is of type $\{2^{k_2},3^{k_3-k_2-1}\}$;
\item[(v)] ${\bf T}_{3}$: $X_1$ is of type $\{1^{k_2-n_2-1},2^{n_2}\}$, and $X_2$ is of type $\{1^{k_1-1},3^{k_3-k_1}\}$.
\end{itemize}

All is left is to consider composing games of the {\bf T} types over a player of level $2$. We start with ${\bf T}_1$. As we know any game of type ${\bf T}_1$ contains a subgame of type ${\bf B}_1$ when we restrict it to lavels 2 and 3 only. For that subgame 2 is the most desirable player so noncompleteness follows from (i). 


Finally we look at ${\bf T}_{3}$ and suppose  now $g$ is a player of level $2$ and $g'$ is a player of level $3$. Here $X_1$ can be taken of type $\{1^{k_2-n_2},2^{n_2-1}\}$. Indeed, if we add $g$ to $X_1$ it becomes winning but it loses with $g'$. Then $X_2$ can be taken of type $\{1^{k_1},2^{k_3-k_1-n_3},3^{n_3-1}\}$. We can add $g$ to $X_2$ since $n_2 \geq k_3-k_1-n_3+1$ and it becomes winning. We can add $g$ and remove $g'$ from it since $n_3\ge 2$. $X_2$ will remain losing after that. So we can again apply Lemma~\ref{X_1X_2} to conclude that the composition is not complete.
This completes the study of compositions where $G_2$ is the anti-unanimity game $A_n$, such that the compositions are not over the least desirable level of $G_1$. 
\end{proof}

Finally, we consider compositions where $G_2$ is the unanimity game $U_n$. It turns out that none of these compositions give a weighted game either, which is what we show next. 

\begin{lemma}
\label{tes_1}
Let $G_1=(P,W)$ be a simple game of one of the types ${\bf B}_1$, ${\bf B}_2$, ${\bf B}_3$, ${\bf T}_{1}$, and ${\bf T}_{3}$ and let $g\in P$ be a player not from the least desirable level of $G_1$. Then the composition $G = G_1 \circ_g U_n$ is not weighted.
\end{lemma}

\begin{proof}
Let $U_n$ be defined on $P_{U_n}$, and let $Z = P_{U_n}$. We start with $G_1$ being of type ${\bf B}_1$. A shift-minimal winning coalition of $G_1$ has the only form $\{1^{k_1},2^{k_2-k_1}\}$, where $k_1 < n_1$. We compose over level $1$ of $G_1$. 
Then $G$ is nonweighted by Lemma~\ref{keylemma} applied to the following trading transform 
\[
(\{1^{k_1}, 2^{k_2-k_1}\}, \{1^{k_1-1}, 2^{k_2-k_1}\}; \{1^{k_1}, 2^{k_2-k_1-1}\}, \{1^{k_1-1}, 2^{k_2-k_1+1}\}).
\]
This is because the first coalition is winning, the second coalition is $1$-winning and the remaining two are losing.
Note that $k_2-k_1+1 = n_2 \geq 2$ in a game of type \textbf{B$_1$}, so the coalition $\{1^{k_1-1}, 2^{k_2-k_1+1}\}$ is allowed. 
\par
Now let $G_1$ be of type ${\bf B}_2$. The shift-minimal winning coalitions of $G_1$ here are $\{1^{k_1}\}, \{2^{k_1+1}\}$, and if we compose with $U_n$ over level $1$ of $G_1$, then $G$ is nonweighted by Lemma~\ref{keylemma} applied to the following trading transform:
\[
( \{2^{k_1+1}\}, \{1^{k_1-1}\} ; \{1^{k_1-1},2\}, \{2^{k_1}\}).
\]
This is because the first coalition is winning and the second is $1$-winning. The remaining two are losing.

Now let $G_1$ be of type ${\bf B}_3$. Recall that in a game of type ${\bf B}_3$ we have $k_1 \leq n_1$, and also $k_2-n_2 < k_1$. So the shift-minimal winning coalitions of $G_1$ are $\{1^{k_1}\}, \{1^{k_2-n_2},2^{n_2}\}$. If we compose with $U_n$ over level $1$ of $G_1$, then $G$ is nonweighted by Lemma~\ref{keylemma} applied to the following trading transform:
\[
(\{1^{k_2-n_2},2^{n_2}\}, \{1^{k_1-1}\}; \{1^{k_2-n_2},2^{n_2-1}\}, \{1^{k_1-1},2\}).
\]
This is because the second coalition is $1$-winning.

Next we look at the games ${\bf T}_1$, and ${\bf T}_{3}$. Since they have three levels each, then we need to consider what happens when composing over level $1$ and when composing over level $2$ separately. Let us start with ${\bf T}_1$.
\par
The shift-minimal winning coalitions of $G_1$ are $\{1^{k_1}\}$ and $\{2^{k_2},3^{k_3-k_2}\}$. Here we need to consider two compositions, one over level $1$, and one over level $2$.
\newline
Case (i). If we compose with $U_n$ over level $1$ of $G_1$ then $G$ is nonweighted by Lemma~\ref{keylemma} applied to the following trading transform:
\[
(\{1^{k_1-1}\}, \{2^{k_2},3^{k_3-k_2}\}; \{1^{k_1-1},2\}, \{2^{k_2-1},3^{k_3-k_2}\}).
\]
This is because the first coalition is $1$-winning, the second is winning and the remaining two are losing.

Case (ii). If we compose with $U_n$ over level $2$ of $G_1$, then $G$ is nonweighted by Lemma~\ref{keylemma} applied to the following trading transform:
\[
(\{1^{k_1}\}, \{2^{k_2-1},3^{k_3-k_2}\}; \{1^{k_1-1},2\}, \{1,2^{k_2-2},3^{k_3-k_2}\}).
\]
This is because the first coalition is winning, the second coalition is $2$-winning and the remaining two are losing.



Finally, let $G_1$ be of type ${\bf T}_{3}$. The shift-minimal winning coalitions of $G_1$ are $\{1^{k_2-n_2},2^{n_2}\}$ and $\{1^{k_1},2^{k_3-k_1-n_3},3^{n_3}\}$. Here we again need to consider two compositions, one over level $1$, one over level $2$.\par

Case (i). If we compose $G_1$ with $U_n$ over level $1$ of $G_1$, then since $k_1\le n_1$, the game $G$ is nonweighted by Lemma~\ref{keylemma} applied to the following trading transform:
\[
(\{1^{k_1},2^{k_3-k_1-n_3},3^{n_3}\}, \{1^{k_1-1},2^{k_3-k_1-n_3},3^{n_3}\}; \{1^{k_1},2^{k_3-k_1-n_3-1},3^{n_3}\}, \{1^{k_1-1},2^{k_3-k_1-n_3+1},3^{n_3}\}).
\]
This is because the first coalition is winning, the second coalition is $1$-winning and the two remaining ones are losing.
Note that $k_3-k_1-n_3+1 \leq n_2$ in a game of type ${\bf T}_{3}$ (see Theorem~\ref{FP2010}), so the last coalition exists. 

Case (ii).  If we compose with $U_n$ over level $2$ of $G_1$, then $G$ is nonweighted by Lemma~\ref{keylemma} applied to the following trading transform: 
\[
(\{1^{k_2-n_2},2^{n_2-1}\}, \{1^{k_1}, 3^{k_3-k_1}\}; \{1^{k_2-n_2},2^{n_2-1},3\}, \{1^{k_1},  3^{k_3-k_1-1}\}).
\]
Indeed, by \eqref{delta_cond_2} $k_2-n_2\le n_1$ and $k_2<k_3$. Thus  the first coalition exists and is $2$-winning, the second is winning and the remaining two are losing.
\end{proof}

We see that none of the six games above produce a weighted game when composed with $U_n$ over a player not from the least desirable level of the first game.

\end{document}